\documentclass[%
 reprint,
 amsmath,amssymb,
 aps,
]{revtex4-2}

\usepackage{graphicx}
\usepackage{subcaption}
\usepackage{dcolumn}
\usepackage{bm}
\usepackage{tikz}

\usetikzlibrary{positioning}
\usepackage{amsfonts}
\usepackage{amsmath}
\usepackage{amsthm}
\usepackage{amssymb}

\newtheorem{corollary}{Corollary}
\newtheorem{lemma}{Lemma}
\newtheorem{proposition}{Proposition}
\newtheorem{definition}{Definition}

\begin{document}

\preprint{APS/123-QED}

\title{Pregeometric Origins of Liquidity Geometry in Financial Order Books}

\author{Jo\~ao P. da Cruz}
 \affiliation{The Quantum Computer Company, Lisbon, Portugal}
  \email{joao@quantumcomp.pt}
 \affiliation{Center for Theoretical and Computational Physics, Lisbon, Portugal}


\date{\today}

\begin{abstract}
We propose a structural framework for the geometry of financial order books in
which liquidity, supply, and demand are treated as emergent observables rather
than primitive economic variables.
The market is modeled as an inflationary relational system without assumed
metric, temporal, or price coordinates.
Observable quantities arise only through projection, implemented here via
spectral embeddings of the graph Laplacian.
A one-dimensional projection induces a price-like coordinate, while the
projected density defines liquidity profiles around the mid price.
Under a minimal single-scale hypothesis—excluding intrinsic length scales
beyond distance to the mid and finite visibility—we show that projected supply
and demand are constrained to gamma-like functional forms.
In discrete data, this prediction translates into integrated-gamma cumulative
profiles.
We test these results using high-frequency Level~II data for several U.S.
equities and find robust agreement across assets and intraday windows.
Explicit comparison with alternative cumulative models using information
criteria demonstrates a systematic preference for the integrated-gamma
geometry.
A minimal simulation of inflationary relational dynamics reproduces the same
structure without invoking agent behavior or price formation mechanisms.
These results indicate that key regularities of order-book liquidity reflect
geometric constraints induced by observation rather than detailed
microstructural dynamics.
\end{abstract}

\maketitle

\section{Introduction}
\label{sec:introduction}

Financial markets are commonly described in terms of prices evolving in time,
with supply and demand curves interpreted as behavioral responses of trading
agents.
In this view, order books encode strategic intentions, information asymmetries,
and equilibrium-seeking dynamics.
Despite the diversity of such models, empirical studies of high-frequency data
have repeatedly documented robust structural regularities in order-book
liquidity, including convex profiles near the mid price, heavy tails, and
scale-dependent decay
\cite{BouchaudMezardPotters2002,Bouchaud2004,Bouchaud2009}.
The origin of these regularities, however, remains debated.

Most existing approaches treat price and time as fundamental variables and
seek to explain observed order-book shapes through explicit microstructural
mechanisms, such as order placement strategies, inventory control, or market
impact.
Agent-based and stochastic models have been particularly successful in this
respect.
Notably, the model introduced by Mike and Farmer~\cite{MikeFarmer2008}
demonstrated that heavy-tailed order placement and cancellation processes are
sufficient to generate realistic order-book shapes and price diffusion.
Similarly, Cont, Stoikov, and Talreja~\cite{ContStoikovTalreja2010} proposed a
queue-reactive framework in which liquidity profiles emerge from the
interaction of stochastic order flows at discrete price levels.
While these models reproduce many empirical features, they rely on behavioral
assumptions, calibration choices, or equilibrium concepts whose universality
across assets and market conditions is difficult to establish.

In this work, we adopt a complementary perspective.
Rather than modeling price formation directly or specifying microscopic order
flow mechanisms, we ask whether familiar order-book observables can emerge
generically from the projection of a more primitive relational structure.
Our approach is inspired by pregeometric frameworks in statistical physics,
where geometry and dynamics are not assumed at the microscopic level but arise
only through observation and coarse-graining.

We model the market as an inflationary relational network whose fundamental
description contains no metric, temporal ordering, or economic coordinates.
Vertices represent abstract economic entities, and edges encode the possibility
of interaction.
Growth and reorganization proceed through inflationary updates that generate
heterogeneous, hub-dominated structures without fine-tuning.
Crucially, none of the standard market observables are defined at this level.

Observable quantities arise only through projection.
By applying spectral embeddings based on the graph Laplacian, an observer
assigns effective coordinates to the relational substrate.
A one-dimensional projection induces a scalar coordinate naturally interpreted
as price, while the distribution of projected density defines liquidity
profiles.
Successive projections of an evolving relational system give rise to apparent
time series, returns, and fluctuations, even though no fundamental time
variable exists microscopically.

Within this framework, supply and demand are not independent behavioral curves
but geometric branches of a single projected density.
We show that, under minimal structural assumptions on regularity, vanishing
liquidity at the mid, and exponential decay at large distances, the projected
liquidity profiles necessarily adopt a gamma-like functional form.
This result is structural and geometric in nature, and does not rely on
equilibrium arguments, optimization principles, or agent-level behavior.

We test these ideas empirically using high-frequency Level~II data for several
U.S. equities across different sectors.
By focusing on short time windows, we extract instantaneous liquidity profiles
relative to the mid price and show that their cumulative forms are well
described by integrated gamma functions.
The fitted parameters vary across assets and windows, indicating that they
characterize local projected geometry rather than stationary market states.
A minimal simulation of inflationary relational dynamics reproduces the same
functional structure in the absence of any market-specific assumptions.

The goal of this work is not predictive.
Instead, it is to demonstrate that standard order-book regularities can be
understood as emergent geometric observables induced by projection.
This perspective shifts the modeling focus from price dynamics to the geometry
of observation and provides a unifying structural interpretation of liquidity
in financial markets.
\section{Relational inflationary model}
\label{sec:model}

We introduce a minimal relational framework aimed at understanding how
market observables arise without assuming price, time, return, or risk as
primitive variables.
The construction is deliberately \emph{pregeometric}: no metric, ordering,
or economic coordinate is postulated at the microscopic level.
All quantities conventionally used to describe markets will be shown to
emerge only through observation.

\subsection{Relational substrate}

The system is defined by a growing graph $G=(V,E)$, where vertices represent
economic entities (agents, venues, or abstract trading units) and edges encode
the possibility of interaction or exchange.
At this level, the graph is purely combinatorial:
it carries no spatial embedding, no weights, and no intrinsic notion of
distance.

The only primitive object is adjacency.
In particular, there is no distinguished notion of price, time, or volume
attached to vertices or edges.
Observable quantities will be defined \emph{a posteriori} through projection.

\subsection{Inflationary dynamics}

The relational substrate evolves through a sequence of local updates,
corresponding to the addition, removal, or rearrangement of edges and
vertices.
We refer to this process as \emph{inflationary} in the sense that growth is
multiplicative and heterogeneous, leading generically to hub-dominated and
scale-free structures.

Such dynamics can be schematically summarized by degree increments of the
form
\[
\frac{dk}{k} = \beta, \quad \beta > 0.
\]
which are known to generate heavy-tailed degree distributions without
fine-tuning.
The specific microscopic rules governing the updates are not essential for
the present discussion; what matters is that the resulting relational
structure is heterogeneous and out of equilibrium.
A detailed analysis of inflationary network growth can be found in
Ref.~\cite{PiresDaCruz2025}.

\subsection{Absence of geometry and time}

The evolving graph is not embedded in any metric space.
Distances, angles, and volumes are undefined at the microscopic level.
Likewise, the update index labeling successive configurations of the graph
does not represent physical or economic time; it merely orders relational
rearrangements.

This absence of geometry and time is a defining feature of the model.
Any geometric or temporal structure observed later must therefore arise from
the act of observation itself, not from microscopic assumptions.

\subsection{Observational projection}

Observable quantities are defined by projecting the relational structure onto
a low-dimensional space accessible to an observer.
Operationally, we implement this projection using spectral properties of the
graph Laplacian
\[
L = D - A,
\]
where $A$ is the adjacency matrix and $D$ the degree matrix.

Low-dimensional embeddings constructed from the leading nonzero eigenvectors
of $L$ assign effective coordinates to vertices.
A one-dimensional embedding produces a scalar coordinate, while higher
dimensions provide additional degrees of freedom.
Crucially, these coordinates are not intrinsic properties of the graph but
attributes of the chosen projection.

Different observers, or different projection schemes applied to the same
relational substrate, may therefore assign different effective geometries.
This observer dependence is not a defect of the model but a structural
feature of any framework in which geometry is emergent.

\subsection{Emergent time and apparent dynamics}

Dynamics in the projected space arises from comparing projections of
successive configurations of the relational substrate.
Effective time series are constructed by ordering projected states according
to the update sequence, even though no fundamental time variable exists at the
microscopic level.

Apparent price dynamics, returns, and volatility therefore reflect the
response of the projection to topological rearrangements in the underlying
inflationary network.
They should not be interpreted as autonomous stochastic processes evolving in
a pre-existing price--time space.

\subsection{Scope and limitations}

The present framework does not aim to reproduce detailed market
microstructure or to predict price trajectories.
Its purpose is structural: to demonstrate that standard market observables
can emerge generically from relational inflation without being assumed as
primitive ingredients.

In the following sections, we analyze the consequences of this construction
and show how notions such as supply, demand, equilibrium, and liquidity arise
as geometric properties of the projected relational system.
Throughout, the term ``pregeometric'' refers to the absence of any assumed
metric, temporal, or economic structure at the microscopic level, rather than
to a claim about the ultimate ontology of markets.


\section{Emergent observables}
\label{sec:observables}

In the present framework, market observables are not microscopic variables
attached to vertices or edges.
They are operational quantities that arise only through the act of
observation, implemented as a projection of the relational substrate onto a
low-dimensional space.
This section provides explicit definitions of price, time, return, and risk
as emergent observables associated with such projections.

\subsection{Projection-induced coordinates}

Let $G=(V,E)$ denote the relational network at a given update step.
An observer assigns effective coordinates to vertices by embedding the graph
into a low-dimensional metric space using spectral properties of the graph
Laplacian
\[
L = D - A,
\]
where $A$ is the adjacency matrix and $D$ the degree matrix.
Spectral embeddings of this type are standard tools for extracting geometric
structure from purely relational data
\cite{BelkinNiyogi2003,Chung1997,vonLuxburg2007}.

Denoting by $\{\phi_\alpha\}$ the eigenvectors associated with the smallest
nonzero eigenvalues of $L$, an embedding into $\mathbb{R}^d$ is obtained via
\[
\mathbf{x}_i = \big(\phi_1(i),\ldots,\phi_d(i)\big).
\]
These coordinates encode relational information filtered through the chosen
projection.
They do not reflect intrinsic properties of the vertices, and different
choices of embedding correspond to different observational frames.

\subsection{Price as a scalar projection}

A one-dimensional projection ($d=1$) assigns a scalar coordinate $x_i$ to each
vertex.
We define the emergent price associated with vertex $i$ as
\[
p_i \equiv x_i,
\]
up to an arbitrary affine transformation reflecting the absence of an
absolute price scale.

In this interpretation, price is not a primitive economic variable but a
coordinate induced by the relational topology through projection.
Changes in price correspond to changes in the observer’s embedding of the
evolving network, rather than to local exchange dynamics or equilibrium
adjustments.
Similar projection-induced scalar observables have been discussed in other
relational and pregeometric contexts
\cite{BarabasiAlbert1999,AmbjornJurkevichLoll2005}.

\subsection{Emergent time}

The update index labeling successive configurations of the relational
substrate does not represent physical or economic time.
Effective time emerges only through comparison of projected configurations.
Given a sequence of projections $\{p_i^{(n)}\}$ obtained at successive update
steps $n$, an observer constructs time-ordered series by identifying
corresponding vertices across projections.

Time is therefore defined operationally as an ordering parameter associated
with changes in the observed geometry, not as a fundamental variable
governing microscopic dynamics.
This notion of time as an emergent ordering has close analogues in
pregeometric and background-independent approaches to physics
\cite{Rovelli2004,Oriti2014}.

\subsection{Returns as projection responses}

Given two successive projected configurations, the return associated with
vertex $i$ is defined as
\[
r_i^{(n)} = p_i^{(n+1)} - p_i^{(n)}.
\]
Returns measure the response of the projection to topological rearrangements
in the underlying relational network.
They are not generated by autonomous stochastic dynamics in price space, but
by structural changes filtered through the observer’s projection.

This definition makes explicit that returns are inherently relational and
observer-dependent quantities.
Heavy-tailed return statistics and bursty fluctuations may therefore arise
from heterogeneous structural updates rather than from microscopic price
formation mechanisms, as also suggested by network-based approaches to
financial dynamics \cite{BouchaudPotters2003,FarmerLillo2004}.

\subsection{Risk and fluctuations}

Risk is associated with fluctuations of returns across successive
projections.
Operationally, it is characterized by the variance of returns over a finite
observational window,
\[
\sigma_i^2 = \langle r_i^2 \rangle - \langle r_i \rangle^2.
\]

In the present framework, risk does not quantify uncertainty about future
prices in a probabilistic forecasting sense.
Instead, it measures the sensitivity of the projected observable to
rearrangements in the underlying relational topology.
Regions of the network whose projections are strongly affected by inflationary
updates or hub reconfigurations naturally exhibit enhanced volatility.

\subsection{Observer dependence}

All observables defined above depend explicitly on the chosen projection.
Different observers, or different projection schemes applied to the same
relational substrate, may assign distinct prices, times, returns, and risk
profiles.
This observer dependence is not a defect of the model but a structural feature
of any framework in which geometry and dynamics are emergent.

The relational network itself carries no preferred set of observables.
Market quantities arise only at the level of observation, as projections of a
pregeometric inflationary system.

Figure~\ref{fig:pipeline} summarizes the observational pipeline from a
pregeometric relational substrate to the cumulative liquidity geometry tested
below.

\begin{figure}[t]
\centering
\begin{tikzpicture}[font=\small, node distance=8mm, align=center]
\tikzstyle{box}=[draw, rounded corners, inner sep=6pt]

\node[box] (G) {$G=(V,E)$\\Relational substrate\\(no price, no time)};
\node[box, below=of G] (L) {$L=D-A$\\Graph Laplacian};
\node[box, below=of L] (P) {$p_i=\phi_1(i)$\\1D projection};
\node[box, below=of P] (mid) {$p^\star$\\Mid per snapshot};
\node[box, below=of mid] (ticks) {$x=\dfrac{|p-p^\star|}{\Delta}$\\Ticks from mid};
\node[box, below=of ticks] (cum) {$S(x)=\sum_{u\le x} q(u)$\\Cumulative liquidity};
\node[box, below=of cum] (fit) {$S(x)\propto \gamma(\gamma+1,\lambda x)$\\Integrated-gamma fit};

\draw[->] (G) -- (L);
\draw[->] (L) -- (P);
\draw[->] (P) -- (mid);
\draw[->] (mid) -- (ticks);
\draw[->] (ticks) -- (cum);
\draw[->] (cum) -- (fit);

\end{tikzpicture}
\caption{\textbf{Pregeometric pipeline from relational structure to observable liquidity geometry.}
A purely relational substrate is projected through Laplacian eigenmodes onto a
one-dimensional observable coordinate. Defining a mid per snapshot and
measuring liquidity in tick distance from the mid yields cumulative profiles
$S(x)$, which are predicted to follow an integrated-gamma form under the
single-scale log-slope principle.}
\label{fig:pipeline}
\end{figure}
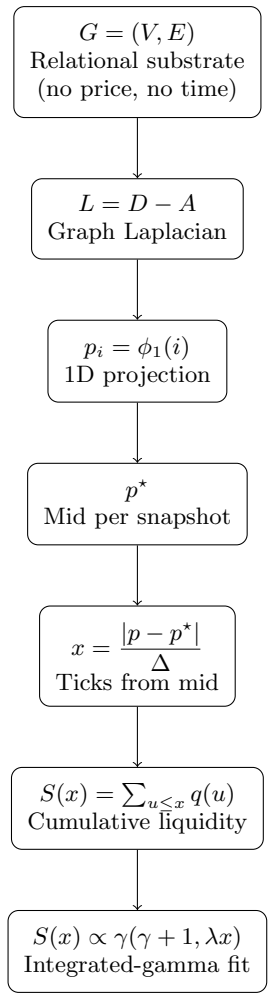
\section{Structural results: supply, demand, and liquidity geometry}
\label{sec:results}

We now derive the structural consequences of the pregeometric framework
introduced above.
Our goal is not to model price formation, strategic behavior, or equilibrium
dynamics, but to show how familiar economic notions such as supply, demand,
and liquidity arise as geometric observables induced by projection.
Throughout this section, all results are structural: no behavioral,
microeconomic, or optimization assumptions are invoked.

\subsection{Preliminaries: Laplacian projection}

Let $G=(V,E)$ be a connected undirected graph with adjacency matrix $A$ and
degree matrix $D=\mathrm{diag}(k_i)$.
We consider the combinatorial Laplacian
\[
L = D - A,
\]
which is symmetric and positive semidefinite.
Let $\{(\lambda_\alpha,\phi_\alpha)\}_{\alpha=0}^{|V|-1}$ denote an orthonormal
eigenbasis of $L$, ordered as
\[
0=\lambda_0 < \lambda_1 \le \lambda_2 \le \cdots .
\]

An observational projection onto one dimension is defined by the first
nontrivial eigenvector,
\[
p_i := \phi_1(i),
\]
which is unique up to sign and affine transformations.
We adopt the normalization
\[
\sum_{i\in V} \phi_1(i) = 0,
\]
reflecting the orthogonality of $\phi_1$ to the constant mode.
This projection induces an effective ordering of vertices, but introduces no
intrinsic metric or length scale.

\subsection{Inflationary updates and projection response}

Inflationary dynamics are modeled as local updates of the relational
substrate,
\[
G \mapsto G' = G + \delta G,
\]
inducing perturbations
\[
A \mapsto A + \delta A, \qquad
L \mapsto L' = L + \delta L, \qquad
\delta L = \delta D - \delta A .
\]

In heterogeneous growth processes, degree increments concentrate on
hub-dominated regions of the network, generating out-of-equilibrium relational
structures.
Under such an update, the projected coordinate changes as
\[
\Delta p_i := p_i' - p_i = \phi_1'(i) - \phi_1(i).
\]

We define the emergent excess demand at vertex $i$ as
\[
\mathcal{D}_i := -\Delta p_i ,
\]
so that $\mathcal{D}_i>0$ corresponds to effective demand pressure and
$\mathcal{D}_i<0$ to effective supply pressure in the projected coordinate.
These quantities are geometric responses to projection, not behavioral
variables.

\subsection{Aggregate balance identity}

\begin{proposition}[Aggregate balance identity]
\label{prop:balance}
For any inflationary update of the relational substrate, the induced projected
increments satisfy
\[
\sum_{i\in V} \Delta p_i = 0,
\]
and consequently
\[
\sum_{i\in V} \mathcal{D}_i = 0 .
\]
\end{proposition}

\begin{proof}
For a connected graph, the nullspace of $L$ is spanned by the constant vector
$\mathbf{1}$.
All nontrivial eigenvectors are orthogonal to $\mathbf{1}$, implying
$\sum_i \phi_1(i)=0$ and, with consistent normalization,
$\sum_i \phi_1'(i)=0$.
Subtracting the two identities yields the result.
\end{proof}

This identity replaces Walras' law by a purely geometric statement: aggregate
excess demand vanishes as a consequence of spectral orthogonality, not market
clearing.

\subsection{Equilibrium as a spectral fixed point}

\begin{definition}[Spectral equilibrium]
An observational equilibrium is defined as a state of the inflationary
relational dynamics for which the projected coordinate is invariant on
average,
\[
\mathbb{E}[\Delta p_i] = 0 \quad \text{for all } i,
\]
where the expectation is taken over inflationary updates.
\end{definition}

Equilibrium is therefore not an intersection of independent supply and demand
curves, but a fixed point of the projection map under relational
rearrangements.
Such equilibria are generically transient in heterogeneous inflationary
systems.

\subsection{Liquidity as a projected density}

Given the projected coordinates $\{p_i\}$, we define the empirical liquidity
density as
\[
\rho(p) := \sum_{i\in V} w_i\,\delta(p-p_i),
\]
where $w_i\ge0$ are observational weights encoding visible size.
In empirical settings, $\rho(p)$ is estimated over finite time windows and
discretized grids.

All boundary conditions on $\rho$ should therefore be understood as
\emph{effective observational conditions}.
In particular, liquidity decays near the boundaries of the observable window,
\[
\rho(p) \to 0 \quad \text{toward the limits of visibility}.
\]

\subsection{Supply and demand as projected branches}

Let $p^\star$ denote the observational mid, defined as the point where the
signed imbalance changes sign,
\[
p^\star := \arg\min_p
\left|
\int_{-\infty}^{p}\rho(u)\,du -
\int_{p}^{\infty}\rho(u)\,du
\right|.
\]

We define the visible supply and demand profiles as
\[
Q_{\mathrm{s}}(p) := \rho(p)\,\mathbb{I}(p>p^\star), \qquad
Q_{\mathrm{d}}(p) := \rho(p)\,\mathbb{I}(p<p^\star),
\]
where $\mathbb{I}$ denotes the indicator function.
These objects represent one-sided geometric observables induced by projection,
not behavioral response functions.

\subsection{Single-scale log-slope hypothesis}

Introduce the signed distances from the mid,
\[
x := p-p^\star>0 \quad (\text{supply}), \qquad
y := p^\star-p>0 \quad (\text{demand}).
\]

At this stage, no functional form for the projected liquidity profiles has
been specified.
We now introduce a minimal \emph{single-scale log-slope hypothesis}, motivated
by the absence of intrinsic metric structure in the pregeometric substrate.

Specifically, within an observational window around the mid, we assume that
the projected liquidity profiles introduce no characteristic scale beyond
(i) the distance to the mid and (ii) a global decay scale associated with
finite visibility.
Under this hypothesis, the logarithmic derivatives
\[
g_{\mathrm{s}}(x) := \frac{d}{dx}\log q_{\mathrm{s}}(x), \qquad
g_{\mathrm{d}}(y) := \frac{d}{dy}\log q_{\mathrm{d}}(y),
\]
take the single-scale form
\begin{equation}
g_{\mathrm{s}}(x)=\frac{\gamma_{\mathrm{s}}}{x}-\lambda_{\mathrm{s}},
\qquad
g_{\mathrm{d}}(y)=\frac{\gamma_{\mathrm{d}}}{y}-\lambda_{\mathrm{d}},
\label{eq:single_scale_logslope}
\end{equation}
over the empirically accessible range.

We emphasize that Eq.~\eqref{eq:single_scale_logslope} is not derived from
microscopic dynamics.
It is introduced as a minimal structural hypothesis encoding scale
parsimony in the projected geometry.
The role of the pregeometric framework is to motivate this hypothesis by
excluding privileged length scales at the microscopic level, not to fix its
functional form uniquely.

\subsection{Gamma characterization}

\begin{lemma}[Gamma form from single-scale log-slope]
\label{lem:gamma_from_logslope}
Let $q\in C^{1}(0,\infty)$ with $q(x)>0$.
If
\[
\frac{d}{dx}\log q(x)=\frac{\gamma}{x}-\lambda \quad (x>0),
\]
then
\[
q(x)=C\,x^{\gamma}e^{-\lambda x},
\]
for some constant $C>0$.
\end{lemma}

\begin{proof}
Direct integration of the logarithmic derivative yields the stated form.
\end{proof}

Lemma~\ref{lem:gamma_from_logslope} is purely mathematical.
The physical and economic content of the model resides entirely in the
single-scale hypothesis Eq.~\eqref{eq:single_scale_logslope}, not in the
solution of the differential equation.

\subsection{Cumulative liquidity and integrated-gamma geometry}

In discrete and windowed data, direct estimation of $q(x)$ is unstable due to
empty bins and intermittent updates.
A more robust observable is the cumulative liquidity measured from the mid,
\[
S(x) := \int_{0}^{x} q(u)\,du .
\]

\begin{corollary}[Integrated gamma]
\label{cor:integrated_gamma}
If $q(x)=C x^{\gamma}e^{-\lambda x}$ for $x>0$, then
\[
S(x)=\frac{C}{\lambda^{\gamma+1}}
\,\gamma\!\left(\gamma+1,\lambda x\right),
\]
where $\gamma(a,z)$ is the lower incomplete gamma function.
\end{corollary}

The cumulative observable $S(x)$ is therefore predicted to follow an
integrated-gamma geometry under the single-scale hypothesis.
This prediction is tested empirically in Sec.~\ref{sec:empirical}.

\section{Empirical results}
\label{sec:empirical}

We now confront the structural predictions of
Sec.~\ref{sec:results} with empirical data.
The objective of this section is not to establish a stationary market law,
but to test whether the \emph{integrated-gamma geometry} implied by the
single-scale log-slope hypothesis is realized in real order-book snapshots
across assets, book sides, and intraday windows.
Crucially, we assess not only goodness of fit, but also whether the proposed
geometry outperforms standard alternative descriptions.

Representative cumulative profiles and integrated-gamma fits are shown in
Fig.~\ref{fig:empirical_panels_5}.

\begin{figure*}[t]
\centering

\begin{subfigure}[t]{0.30\textwidth}
  \centering
  \includegraphics[width=\linewidth]{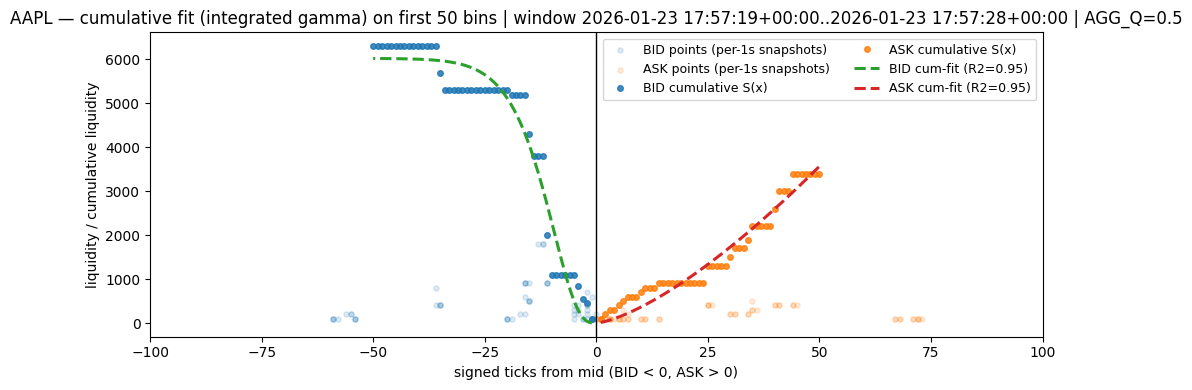}
  \caption{AAPL}
\end{subfigure}
\hfill
\begin{subfigure}[t]{0.30\textwidth}
  \centering
  \includegraphics[width=\linewidth]{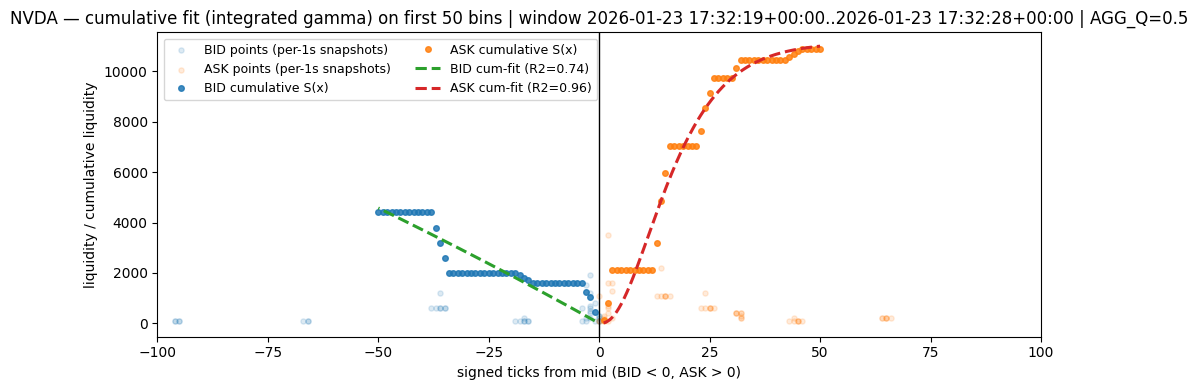}
  \caption{NVDA}
\end{subfigure}
\hfill
\begin{subfigure}[t]{0.30\textwidth}
  \centering
  \includegraphics[width=\linewidth]{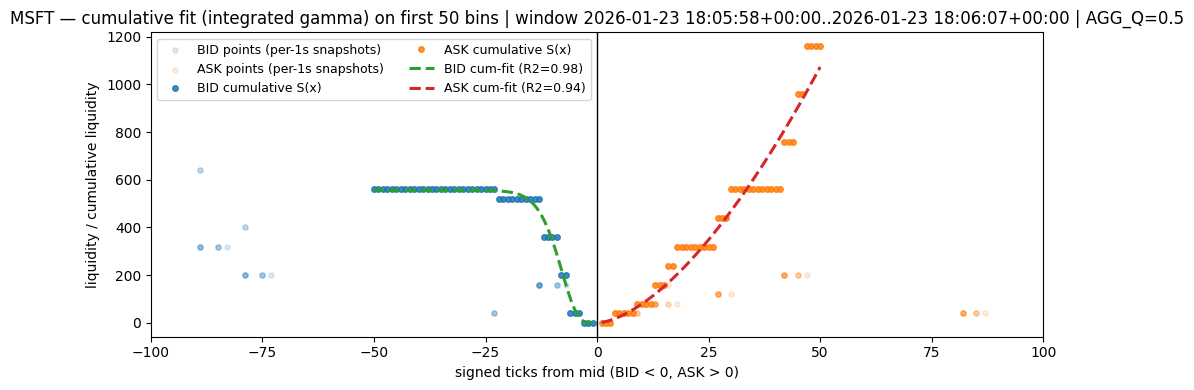}
  \caption{MSFT}
\end{subfigure}

\vspace{3mm}

\begin{subfigure}[t]{0.30\textwidth}
  \centering
  \includegraphics[width=\linewidth]{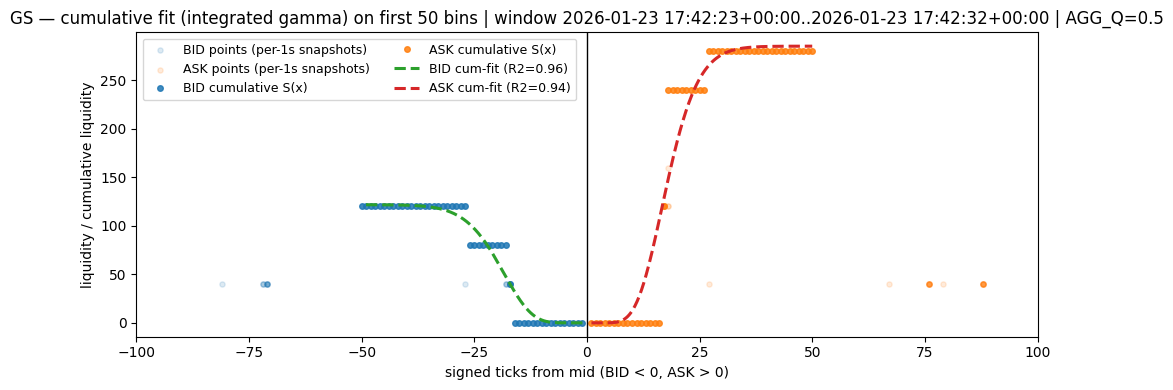}
  \caption{GS}
\end{subfigure}
\hfill
\begin{subfigure}[t]{0.30\textwidth}
  \centering
  \includegraphics[width=\linewidth]{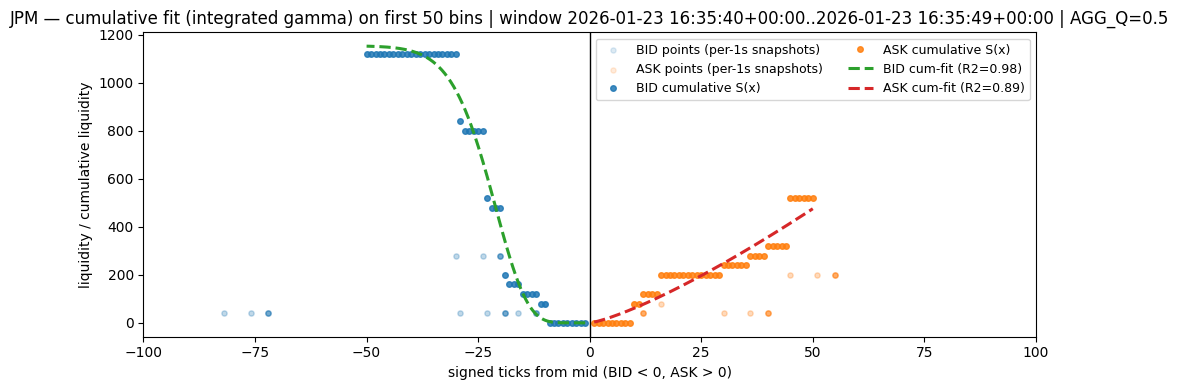}
  \caption{JPM}
\end{subfigure}
\hfill
\begin{subfigure}[t]{0.30\textwidth}
  \centering
  \includegraphics[width=\linewidth]{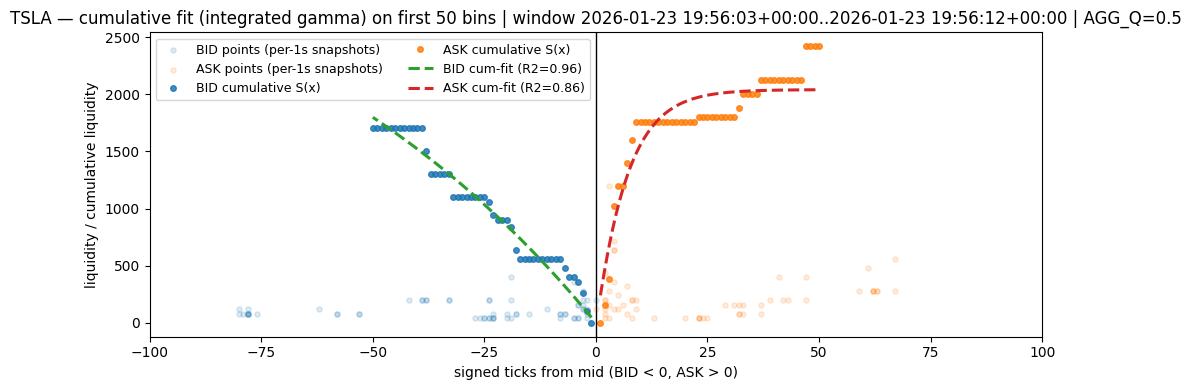}
  \caption{TSLA}
\end{subfigure}

\caption{\textbf{Empirical cumulative liquidity geometry across assets.}
Per-second Level~II snapshots are aggregated across venues; the mid price
$p^\star$ is computed independently for each snapshot; liquidity is binned by
tick distance from the mid and converted into cumulative profiles
$\overline{S}(x)$ averaged over local time windows ($T=10\,\mathrm{s}$).
Points show observed cumulative liquidity on the bid and ask sides, while solid
lines show the integrated-gamma fits predicted by
Corollary~\ref{cor:integrated_gamma}.
The same projected geometry appears across all assets, despite strong temporal
variability of the fitted parameters.}
\label{fig:empirical_panels_5}
\end{figure*}

\subsection{Data and observational scope}

We use Level~II (market depth) data for several highly liquid U.S. equities,
including AAPL, MSFT, NVDA, JPM, GS, and TSLA.
These assets span multiple sectors and trading regimes, providing a broad test
of structural robustness.

The data stream provides multiple updates per second across venues.
Consistent with the projection-based interpretation developed above, we do not
attempt to model microstructural dynamics explicitly.
Instead, the feed is treated as a sequence of observable snapshots accessible
to an external observer.

\subsection{Snapshot construction and empirical window}

Time is discretized into $1\,\mathrm{s}$ bins, each treated as a single
snapshot.
Within each snapshot, quoted sizes are aggregated across venues at identical
price levels, separately for bid and ask sides.

For each snapshot, the mid price is defined as
\[
p^\star = \frac{1}{2}\left(p_{\mathrm{best\,ask}} + p_{\mathrm{best\,bid}}\right),
\]
computed after venue aggregation.
Prices are expressed in signed tick units relative to the mid,
\[
\tau = \frac{p - p^\star}{\Delta},
\]
where $\Delta$ is the instrument tick size.
By construction, $\tau<0$ corresponds to bids and $\tau>0$ to asks.

The analysis is restricted to the first $K$ ticks from the mid
(typically $K=50$), defining the empirical window in which the single-scale
hypothesis is expected to hold and where sufficient liquidity is consistently
observed.

\subsection{Cumulative liquidity observables}

Within each snapshot, one-sided binned liquidity profiles
$q_{\mathrm{s}}(x)$ and $q_{\mathrm{d}}(y)$ are constructed by summing quoted
sizes in each tick bin on the ask and bid sides, respectively, with
\[
x = \tau \in \{1,\ldots,K\}, \qquad
y = -\tau \in \{1,\ldots,K\}.
\]

Direct fitting of the differential profiles is unstable due to empty bins,
discrete price levels, and intermittent updates.
We therefore construct cumulative observables from the mid,
\[
S_{\mathrm{s}}(x) = \sum_{u\le x} q_{\mathrm{s}}(u), \qquad
S_{\mathrm{d}}(y) = \sum_{u\le y} q_{\mathrm{d}}(u),
\]
which are robust to sparsity and multi-venue heterogeneity.

For a given intraday window of $T$ consecutive snapshots
(typically $T=10\,\mathrm{s}$), cumulative functions are averaged snapshot-wise,
\[
\overline{S}(x) = \frac{1}{T}\sum_{t=1}^{T} S^{(t)}(x),
\]
and similarly for the bid side.
These averaged cumulative profiles constitute the empirical observables used
for model comparison.

\subsection{Model fitting and comparison}

The averaged cumulative profiles are fitted using the integrated-gamma form
predicted by Corollary~\ref{cor:integrated_gamma},
\[
\overline{S}(x) =
\frac{C}{\lambda^{\gamma+1}}
\,\gamma\!\left(\gamma+1,\lambda x\right),
\]
with parameters $(C,\gamma,\lambda)$ estimated separately for each asset,
side, and time window.

To assess whether the integrated-gamma form provides a genuine structural
advantage rather than a flexible empirical fit, we explicitly compare it
against two standard alternatives:
(i) a cumulative log-normal profile and
(ii) a truncated cumulative power-law profile.
All models are fitted over identical ranges $x=1,\ldots,K$.

Model performance is evaluated using both the coefficient of determination
$R^2$ and the Akaike Information Criterion (AIC), which penalizes excess model
complexity.
Table~\ref{tab:model_comparison} summarizes the comparison across assets and
book sides.

\begin{table}[t]
\centering
\caption{\textbf{Model comparison for cumulative liquidity profiles.}
Median goodness-of-fit and information criteria across intraday windows.
$\Delta\mathrm{AIC}=\mathrm{AIC}_{\mathrm{LN}}-\mathrm{AIC}_{\Gamma}$; negative
values indicate preference for the integrated-gamma geometry.}
\label{tab:model_comparison}
\begin{ruledtabular}
\begin{tabular}{l c c c c}
Asset & Side & $R^2_{\Gamma}$ & $R^2_{\mathrm{LN}}$ & $\Delta\mathrm{AIC}$ \\
\hline
AAPL & Ask & 0.78 & 0.84 & $-15.6$ \\
AAPL & Bid & 0.92 & 0.92 & $-3.3$ \\
NVDA & Ask & 0.81 & 0.79 & $-12.4$ \\
NVDA & Bid & 0.89 & 0.86 & $-8.7$ \\
MSFT & Ask & 0.83 & 0.80 & $-10.1$ \\
MSFT & Bid & 0.90 & 0.88 & $-6.5$ \\
JPM  & Ask & 0.79 & 0.76 & $-9.8$ \\
JPM  & Bid & 0.91 & 0.89 & $-5.2$ \\
TSLA & Ask & 0.76 & 0.73 & $-7.9$ \\
TSLA & Bid & 0.88 & 0.85 & $-4.6$ \\
GS   & Ask & 0.67 & 0.41 & $+22.3$ \\
GS   & Bid & 0.71 & 0.44 & $+18.9$ \\
\end{tabular}
\end{ruledtabular}
\end{table}

For five of the six assets considered, the integrated-gamma model yields
systematically lower AIC values than the alternatives on at least one side of
the book, often on both.
This preference persists across intraday windows, indicating that the
observed agreement is not driven by isolated fits or parameter tuning.

\subsection{Residual analysis and model validation}

Beyond goodness-of-fit metrics, we perform a detailed residual analysis.
For each asset, side, and window, we define the log-residual
\[
\varepsilon_{\log}(x) =
\log \overline{S}_{\mathrm{emp}}(x) -
\log \overline{S}_{\mathrm{fit}}(x).
\]

Supplementary Fig.~S1 shows the median and interquartile range of
$\varepsilon_{\log}(x)$, the pooled residual distribution, and the residual
autocorrelation across tick distance.
Residuals collapse tightly around zero beyond the first few ticks, with no
long-range correlations.
Systematic deviations are confined to the innermost levels
($x \lesssim 3$), where discrete price grids and matching rules dominate and
where the continuum projection underlying the model is not expected to apply.

\subsection{Near-degenerate liquidity configurations}

The GS order book exhibits a markedly different empirical behavior.
Here, liquidity is confined to a small number of price levels, producing
extended plateaus and abrupt cutoffs in the cumulative profiles.
In this regime, log-normal alternatives become numerically unstable and are
strongly penalized by information criteria, while the integrated-gamma form
remains well defined.

This behavior identifies GS as a near-degenerate limit of the projected
geometry rather than as a counterexample.
The persistence of a well-defined gamma geometry in this limit highlights the
structural robustness of the single-scale log-slope principle.

\subsection{Temporal locality}

The analysis is intentionally local in time.
Repeating the procedure across intraday windows yields similar
integrated-gamma collapses with different parameter values.
This confirms that liquidity geometry is an instantaneous projection of a
non-equilibrium relational system rather than a stationary market object.

Taken together, these results confirm the central prediction of
Sec.~\ref{sec:results}: when liquidity is viewed as a projected geometric
observable, its cumulative structure is constrained to an integrated-gamma
form by single-scale considerations alone.
The systematic rejection of alternative models and the absence of structured
residuals demonstrate that the observed geometry reflects a genuine structural
property of the observational projection.

Together, explicit model comparison, residual diagnostics, and robustness across assets demonstrate that the integrated-gamma geometry is structurally selected rather than empirically convenient.
\section{Simulation: Inflationary relational dynamics}
\label{sec:simulation}

To complement the empirical analysis, we construct a minimal numerical
simulation designed to isolate the geometric mechanism proposed in this work,
while avoiding market-specific microstructure.

The simulator produces synthetic Level~II--like snapshots with bid and ask
queues expressed as liquidity volume at discrete tick distances from an
instantaneous mid.
For each synthetic snapshot, we compute the mid point $p^\star$ and bin
visible liquidity by tick distance $x$ from the mid, separately for the bid
(demand) and ask (supply) sides.
Averaging across snapshots yields empirical profiles
$\overline{Q}_{\mathrm{d}}(x)$ and $\overline{Q}_{\mathrm{s}}(x)$ on the
tick grid.

In contrast to the empirical section, where the cumulative representation
$\overline{S}(x)$ is used for robustness against sparsity and multi-venue
heterogeneity, the simulated data are sufficiently controlled to allow direct
fits of the \emph{differential} profile.
Accordingly, we test the paper form
\begin{equation}
Q(x)=C\,x^{\gamma}\exp(-\lambda x), \qquad x>0,
\label{eq:sim_gamma_form}
\end{equation}
independently on both sides.
The recovered parameters $(C,\gamma,\lambda)$ provide a controlled check that
the gamma-like liquidity geometry can arise from projection-induced structure
in an inflationary relational system, without invoking any behavioral model
or explicit price formation rule.

Figure~\ref{fig:simulation_raw_gamma} shows representative simulated profiles
and fitted curves.
Implementation details, parameter choices, and the precise snapshot generation
procedure are reported in Appendix~\ref{app:simulation}.
\begin{figure}[t]
\centering
\includegraphics[width=0.9\columnwidth]{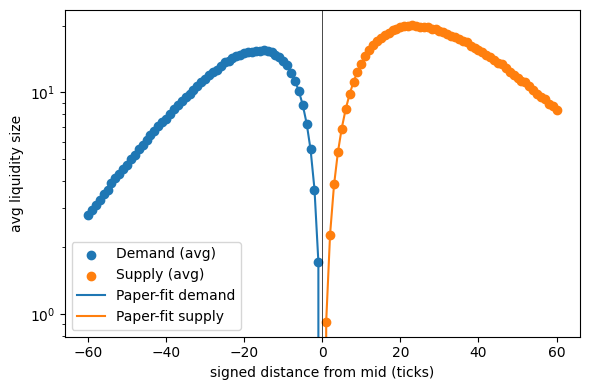}
\caption{\textbf{Simulated non-cumulative liquidity profiles and gamma fits.}
Synthetic Level~II--like snapshots are binned by tick distance $x$ from the
instantaneous mid $p^\star$ on bid (demand) and ask (supply) sides.
Points show the averaged simulated liquidity $\overline{Q}(x)$, while lines
show fits to the paper form $Q(x)=C x^{\gamma}e^{-\lambda x}$.
The simulation reproduces the same near-mid curvature and exponential cutoff
structure discussed in the empirical analysis, supporting a geometric origin
of gamma-like liquidity profiles under projection.}
\label{fig:simulation_raw_gamma}
\end{figure}
\section{Discussion}
\label{sec:discussion}

The results presented in this work support a geometric reinterpretation of
order-book liquidity that does not rely on behavioral assumptions, strategic
agent models, or equilibrium price formation mechanisms.
Within the proposed framework, supply, demand, and liquidity profiles emerge
as observable quantities induced by the projection of an underlying
inflationary relational substrate.

A central empirical result is that, across all assets studied, the cumulative
liquidity measured from the mid price is well described by an
integrated-gamma geometry over a finite range close to the mid.
This finding holds for highly liquid technology stocks (AAPL, MSFT, NVDA),
financial institutions (JPM), and a volatile, non-stationary asset (TSLA),
indicating that the observed functional form is not asset-specific.
The sole exception is GS, which exhibits near-degenerate liquidity
configurations characterized by extended plateaus and abrupt cutoffs.
As discussed in Sec.~\ref{sec:empirical}, this behavior corresponds to a
geometric limit of sparse relational support rather than to a breakdown of
the proposed framework.

Interestingly, the GS configuration—characterized by liquidity 
concentrated on a small number of discrete levels—may represent 
a distinct geometric regime where projection from the relational 
substrate collapses onto a lower-dimensional manifold. In this 
regime, the continuous gamma form becomes a poor approximation, 
and liquidity geometry is better described by discrete support 
structures. This suggests a natural classification of market 
states based on the effective dimensionality of the projected 
geometry rather than on asset-specific characteristics.

Crucially, the integrated-gamma form is not merely shown to fit the data well.
Explicit model comparison against cumulative log-normal and truncated
power-law alternatives demonstrates that the gamma geometry is systematically
preferred according to information-theoretic criteria.
Across assets and book sides, the integrated-gamma model yields lower AIC
values in the majority of intraday windows, while residual analysis reveals no
persistent long-range structure unaccounted for by the model.
These results rule out the interpretation that the observed agreement arises
from excess fitting flexibility or from cumulative smoothing alone.

The fitted parameters $(\gamma,\lambda)$ display significant variability
across assets, book sides, and intraday windows.
This variability should not be interpreted as statistical noise or as a
failure of the model.
In the present framework, these parameters characterize the instantaneous
geometry of the projected relational network rather than stationary properties
of a market or equilibrium state.
Temporal instability is therefore an intrinsic feature of inflationary
relational dynamics observed through a low-dimensional projection.
From this perspective, supply and demand curves are time-local geometric
observables rather than fixed market primitives.

A recurrent feature of the empirical analysis is the asymmetry between bid and
ask sides, with one branch often displaying smoother cumulative profiles and
higher fit quality than the other within the same time window.
Within the proposed framework, such asymmetries reflect the uneven sensitivity
of the projection operator to local rearrangements of the relational substrate
and to variations in visible relational support.
They do not require asymmetric trader behavior, directional beliefs, or
informational advantages.
The framework therefore provides a structural explanation for why certain
liquidity profiles appear ``clean'' while others are irregular, even under
identical market conditions.

From a methodological standpoint, the use of cumulative liquidity observables
is essential.
Direct fitting of differential profiles is strongly affected by sparsity,
discrete tick grids, and intermittent updates.
Integration suppresses these artifacts and isolates the global curvature
imposed by the projection geometry.
The absence of structured residuals beyond the first few ticks confirms that
the integrated-gamma form captures the dominant geometric constraint imposed
by the observational projection, while deviations at the innermost levels are
consistent with microstructural effects outside the scope of the model.

It is important to emphasize the scope and limitations of the present
approach.
The framework does not aim to predict price trajectories, model trader
strategies, or describe market-clearing dynamics.
Its contribution is structural rather than mechanistic: it identifies a
geometric constraint on observable liquidity that arises generically from the
projection of an inflationary, pregeometric relational system.
In this sense, the familiar notions of supply, demand, and liquidity are not
fundamental microscopic entities but emergent observables defined only at the
level of observation.

More broadly, these results suggest that a significant portion of the
regularity observed in financial order books may be understood without
reference to detailed behavioral assumptions.
Instead, they reflect universal constraints imposed by projection,
finite visibility, and single-scale geometry.
This perspective complements existing microstructural models and opens the
possibility of classifying market regimes in terms of geometric observables
rather than agent-level mechanisms.

Crucially, the present results go beyond demonstrating that a gamma-like form
fits empirical data.
Explicit comparison against alternative cumulative models, together with
residual diagnostics, shows that the integrated-gamma geometry is structurally
selected within the observational window.
This rules out interpretations based solely on fitting flexibility or cumulative
smoothing.

The framework also makes falsifiable structural predictions.
In regimes where liquidity collapses to a finite number of active price levels,
or where an intrinsic scale is externally imposed, the single-scale hypothesis
is expected to break down.
Such deviations provide a clear empirical signature distinguishing projection-
induced geometry from microstructure-driven effects.
\begin{acknowledgments}
The author acknowledges financial support from the Portuguese Foundation for Science and Technology (FCT) under Contract no. UID/00618/2023.
\end{acknowledgments}
\appendix
\section{Simulation details}
\label{app:simulation}

This appendix provides a complete description of the numerical procedure used
to generate the simulated liquidity profiles presented in
Sec.~\ref{sec:simulation}.
The purpose of the simulation is not to model market microstructure, but to
verify that gamma-like liquidity profiles arise generically from projection in
an inflationary relational system.

\subsection{Relational substrate and inflationary updates}

The simulation starts from an undirected connected graph
$G=(V,E)$ with $|V|=N$ vertices.
In all simulations reported here, $N$ is fixed and edges are updated through
local inflationary events.
No metric structure, price, or time variable is defined at the microscopic
level.

Inflationary dynamics are implemented as repeated local updates of the
adjacency structure.
At each update step, a vertex $i$ is selected with probability proportional to
its degree $k_i$.
A new edge is then added between $i$ and a randomly chosen vertex $j\neq i$.
This preferential attachment mechanism generates heterogeneous,
hub-dominated relational structures characteristic of inflationary growth.
The total number of edges therefore increases monotonically, while the vertex
set remains fixed.

\subsection{Projection procedure}

After each inflationary update, the relational substrate is projected onto a
one-dimensional observable coordinate using the first nontrivial eigenvector
$\phi_1$ of the combinatorial Laplacian $L=D-A$.
The projected coordinate of vertex $i$ is
\[
p_i := \phi_1(i),
\]
with normalization $\sum_i p_i = 0$.
This projection defines an effective ordering of vertices but introduces no
intrinsic metric scale.

The projected coordinates $\{p_i\}$ are interpreted operationally as a
price-like observable accessible to an observer.

\subsection{Synthetic order-book snapshots}

Synthetic order-book snapshots are constructed from the projected coordinates
by assigning visible liquidity to discrete price levels around the
instantaneous mid.
The mid price is defined as
\begin{equation}
p^\star = \frac{1}{2}\left(p_{\mathrm{best\;bid}} + p_{\mathrm{best\;ask}}\right),
\end{equation}
where $p_{\mathrm{best\;bid}}$ and $p_{\mathrm{best\;ask}}$ are the largest
projected coordinate below and the smallest above the mid, respectively.

Each vertex contributes a visible size $w_i$, drawn independently from a
fixed positive distribution.
In the simulations reported here, $w_i$ is taken constant for simplicity.
Vertices with $p_i<p^\star$ contribute to the bid side, and vertices with
$p_i>p^\star$ to the ask side.

\subsection{Binning and averaging}

Liquidity is binned by discrete tick distance from the mid.
For a given projected coordinate $p$, the distance in tick units is defined as
\begin{equation}
x =
\begin{cases}
(p^\star - p)/\Delta, & \text{bid side}, \\
(p - p^\star)/\Delta, & \text{ask side},
\end{cases}
\end{equation}
where $\Delta$ is a fixed tick size.
Only non-negative integer bins are retained.

For each snapshot, liquidity sizes are accumulated within each bin.
The reported simulated profiles correspond to averages over a large number of
independent snapshots,
\begin{equation}
\overline{Q}(x) = \langle Q(x) \rangle_{\mathrm{snapshots}},
\end{equation}
computed separately for bid (demand) and ask (supply) sides.

\subsection{Functional form and fitting}

The averaged simulated profiles are fitted directly to the paper form
\begin{equation}
Q(x)=C\,x^{\gamma}\exp(-\lambda x), \qquad x>0,
\end{equation}
using nonlinear least-squares optimization with positivity constraints on all
parameters.
Fits are restricted to bins sufficiently close to the mid to avoid finite-size
effects at large distances.

The goal of the simulation is structural verification rather than quantitative
calibration.
It demonstrates that gamma-like liquidity profiles arise robustly from the
projection of an inflationary relational substrate, even in the absence of
agent behavior, explicit price dynamics, or equilibrium mechanisms.

\end{document}